\documentclass[letterpaper, 10 pt, conference]{ieeeconf}  

\IEEEoverridecommandlockouts                              

\usepackage[noadjust]{cite} 

\usepackage[utf8]{inputenc} 
\usepackage[T1]{fontenc}    
\usepackage{hyperref}       
\usepackage{url}            
\usepackage{booktabs}       
\usepackage{amsfonts}       
\usepackage{nicefrac}       
\usepackage{microtype}      
\usepackage{xcolor}         

\usepackage{graphicx}
\usepackage{booktabs} 

\usepackage{caption}

\def \basp {0pt}
\def \insp {0pt}

\usepackage{amssymb}
\usepackage{mathtools}

\usepackage{algorithm}
\usepackage{algorithmic}

\usepackage{subcaption}

\usepackage{amsfonts}
\usepackage{bbm}

\usepackage[capitalize,noabbrev]{cleveref}

\usepackage{float}

\usepackage[export]{adjustbox} 

\usepackage{setspace}

\newcommand{\CASE}[1]{\STATE \textbf{case} #1\textbf{:} \begin{ALC@g}}
\newcommand{\ENDCASE}{\end{ALC@g}}

\newcommand{\DEFAULT}{\STATE \textbf{default:} \begin{ALC@g}}
\newcommand{\ENDDEFAULT}{\end{ALC@g}}
\newcommand{\DEFAULTLINE}[1]{\STATE \textbf{default:} }

\usepackage{amssymb}
\usepackage{amsmath}
\usepackage{amsthm}

\newcommand{\RNum}[1]{\uppercase\expandafter{\romannumeral #1\relax}}

\newcommand\norm[1]{\left\lVert#1\right\rVert}
\usepackage{commath} 
\let\labelindent\relax 
\usepackage{enumitem}

\newcounter{casenum}

\newcounter{assnum}

\newcommand{\br}{\mathbb{R}}

\DeclareMathOperator{\Ex}{\mathbb{E}}
\DeclareMathOperator*{\argmax}{arg\,max}

\theoremstyle{plain}

\newtheorem{theorem}{Theorem}[section]

\newtheorem{lemma}[theorem]{Lemma}

\theoremstyle{definition}
\newtheorem{definition}[theorem]{Definition}
\theoremstyle{remark}

\title{\LARGE \bf
Sample-Optimal Zero-Violation Safety For Continuous Control
}

\author{Ritabrata Ray$^{1}$, Yorie Nakahira$^{{1}}$, and Soummya Kar$^{1}$
\thanks{$^{1}$Ritabrata Ray, Yorie Nakahira and Soummya Kar are with the Department of Electrical \& Computer Engineering, 
        Carnegie Mellon University, Pittsburgh, PA 15213, USA
        {\tt\small \{ritabrar, ynakahir, soummyak\} @andrew.cmu.edu}}%
}

\begin{document}

\maketitle
\thispagestyle{empty}
\pagestyle{empty}

\begin{abstract}

In this paper, we study the problem of ensuring safety with a few shots of samples for partially unknown systems. We first characterize a fundamental limit when producing safe actions is not possible due to insufficient information or samples. Then, we develop a technique that can generate provably safe actions and recovery (stabilizing) behaviors using a minimum number of samples. In the performance analysis, we also establish Nagumo's theorem-like results with
relaxed assumptions, which is potentially useful in other contexts. Finally, we discuss how the proposed method can be integrated into the policy gradient algorithm to assure safety and stability with a handful of samples without stabilizing initial policies or generative models to probe safe actions. 

\end{abstract}

\section{Introduction}
\label{sec: Intro}

Assuring zero-violation safety (the system never gets unsafe) in uncertain systems is challenging when sufficient samples or generative models are not available. System identification and construction of generative models usually require a large number of samples to be gathered from exploring physical systems, but unsafe actions during the exploration process can have significant consequences. Existing literature has primarily focused on situations when the system dynamics are known or when the outcome of actions can be probed from a generative model. But these methods still require a certain number of samples before safe actions can be ensured. 

Motivated by this challenge, we investigate the following questions related to safe learning and control problems in unknown dynamics: 
\begin{enumerate}[topsep=\basp,itemsep=0pt,partopsep=0pt, parsep=\insp]
\item Fundamental limits: What are the fundamental limits in ensuring zero-violation safety? 
\item Achievable performance and efficient algorithms: What is the minimal information required to achieve zero-violation safety at the fundamental limits? 
How to design a computational algorithm that realizes this performance?
\item Modular architecture to assist safe control and learning: 
How to integrate the computational algorithm into a nominal (existing) control loop? How to integrate the algorithm into a reinforcement learning policy? 
\end{enumerate}



  
To answer questions 1 and 2, we first show the fundamental limits and develop an algorithm that operates at the limits. 
This method is constructed by re-deriving forward invariance conditions analogous to Nagumo's theorem~(\cite{Nagumo1942berDL}) for right-continuous dynamics and using the conditions to generate safe actions based on samples from an infinitesimal history. We then apply this idea in policy gradient reinforcement learning to ensure safety during exploration and after convergence. The merits of the proposed method are summarized below.  
\begin{enumerate} 
    \item \label{Mr: Safety} \textit{Safety with instantaneous samples:} The proposed technique constructs provably-safe actions using the instantaneous histories of the state and action (Theorems~\ref{thm: sample case optimality},~\ref{thm: main result}, Figures~\ref{fig: Bar Adaptive Control Plot}). 
    \item \label{Mr: Recovery Extension} \textit{Guaranteed recovery speed:} The proposed technique is able to guarantee recovery from unsafe states using only instantaneous histories (Theorem \ref{thm: main result}, Figure~\ref{fig: Bar Recovery Plot}).
    \item \label{Mr: Model-Free RL Tool} \textit{Applicability to policy gradient reinforcement learning:} Our method can be integrated into a policy-gradient learning algorithm. The integrated method will allow safety during exploration (Theorem~\ref{thm: main result}, Figure~\ref{fig: Model-Free Safety Rate}) and learn optimal policies with minimum degradation in convergence and optimal cost (Theorem~\ref{thm: Gradient Unbiased Estimate}, Figure~\ref{fig: Model-Free RL Convergence}). 
\end{enumerate}
The requirement of instantaneous histories for safety and recovery for our method is at the fundamental limit below which zero-violation safety is impossible. 
The proposed method can be modularly combined with nominal control policies as in algorithm~\ref{alg: main algorithm} and applied when a policy is iteratively learned as in algorithm~\ref{alg: REINFORCE with Safety}.



\textbf{Related Work.} Safe decision-making for unknown system dynamics has been studied in several contexts, such as safe control, adaptive control, and reinforcement learning~(\cite{https://doi.org/10.48550/arxiv.2205.10330, JMLR:v16:garcia15a}). 
Many system identification techniques are developed to learn the models and parameters of system dynamics using samples of system trajectories~(\cite{book,GROVER2021346, 9682628}). 
The identified system models are often used in safe control, robust control, and optimization-based control to decide control actions or policies~(\cite{https://doi.org/10.48550/arxiv.1906.11392}). Barrier-function-based techniques are used to characterize the set of safe control actions by constraining the evolution of barrier function values over time~(\cite{9147463,9129764,9196709,9683085,https://doi.org/10.48550/arxiv.2005.07284,https://doi.org/10.48550/arxiv.2010.16001,https://doi.org/10.48550/arxiv.2104.14030}). This approach is also combined with adaptive control techniques to adapt to changing system parameters~(\cite{9147463,9129764,9196709,9683085}). 
Methods based on this approach often require known dynamical systems, bounded or small parametric uncertainties in the dynamics model, or the availability of simulating models from which safe/unsafe actions can be probed.

Online learning is used to balance exploration (to learn the system) and exploitation (to optimize performance) and obtain control policies that achieve sub-linear regret~(\cite {https://doi.org/10.48550/arxiv.2207.14419,LC3}). These methods employ techniques such as uncertainty quantification with optimistic choices, Thompson sampling, or via reduction to combinatorial optimization problems like Convex Body Chasing (\cite{https://doi.org/10.48550/arxiv.2103.11055}). Often using Gaussian Process as priors to quantify the uncertainties ( assuming some regularity in the dynamics) in the system dynamics yields high probability safe exploration guarantees ~(\cite{https://doi.org/10.48550/arxiv.1705.08551,Ma_Shen_Bastani_Dinesh_2022,DBLP:journals/corr/abs-2006-09436,DBLP:journals/corr/abs-1903-02526,https://doi.org/10.48550/arxiv.1712.05556, pmlr-v37-sui15,https://doi.org/10.48550/arxiv.1606.04753}). Such methods usually require heavy computation and initial stabilizing policies.

Many techniques have also been developed in the context of safe reinforcement learning. 
Some methods impose safety requirements in reward functions~(\cite{https://doi.org/10.48550/arxiv.2205.11814, DONG202083}), constraints~(\cite{pmlr-v70-achiam17a,bharadhwaj2021conservative,ijcai2020p632,https://doi.org/10.48550/arxiv.2011.06882,vuong2018supervised,Yang2020Projection-Based}), in the Lagrangian to the cumulative reward objective~(\cite{ray2019benchmarking}).
In these techniques, the safety requirements are often imposed as chance constraints or in expectation, as opposed to hard deterministic constraints. 
Other methods use barrier-function-based approaches in reinforcement learning~(\cite{https://doi.org/10.48550/arxiv.2201.01918, zhao2021modelfree, https://doi.org/10.48550/arxiv.2004.07584,NEURIPS2018_4fe51490,DBLP:journals/corr/abs-1901-10031,DONG202083,DBLP:journals/corr/abs-2002-10126,DBLP:journals/corr/abs-2107-13944,article2,article,zhao2021modelfree}). 
Many of these techniques are concerned with convergence to near-optimal safe policies or learning accurate model parameters with sufficiently many samples.
On the other hand, we consider a complementary problem: what can be achieved during the initial learning phase and when available information is at the fundamental limits.

\textbf{Organization of the paper.} This paper is organized as follows. We present the model and problem statement in section~\ref{sec: system description}; fundamental limits, proposed techniques, and performance guarantees in section~\ref{sec: proposed method}; an application to safe learning in section~\ref{sec: RL applications}; experiments in section~\ref{sec: simulations}; and conclusion in section~\ref{sec: Paper Conclusion}.

\textbf{Notation.} We adopt the following notations. For any differentiable function $f:\br^d\rightarrow \br,$ let $\nabla f$ denotes its gradient with respect to $x$, i.e. $\nabla_x f(x)$. Let $x_t: \br_{\geq 0}\rightarrow \br^d$ denote the state of any dynamical system as a function of time $t$, then $\dot f$ denotes the time derivative of the function $f(x_t)$ with respect to $t$. $\dot f^{+}$ denotes its right hand time derivative, and $\dot f^{-}$ denotes its left hand time derivative: i.e., $\dot f^{+}=\lim_{\delta t \rightarrow 0^+}\frac{f(x_{t+\delta t})-f(x_{t})}{\delta t}$, and $\dot f^{-}=\lim_{\delta t \rightarrow 0^+}\frac{f(x_{t})-f(x_{t-\delta t})}{\delta t}$.
For any two vectors $u,v$ of the same dimensions, $u \cdot v$ or $\langle u, v\rangle$ denote their standard inner product in Euclidean space. We use $\Pr[A]$, $\Pr[A|B]$ to denote probabilities of events $A$, and conditional probability of $A$ given $B$ respectively. Given random variables $X$ and $Y$, $p(x,y)$ denotes their joint density function, and $p(x|y)$ denotes the conditional probability density function at $X=x$, and $Y=y$. 


\section{Model and Problem Statement}
\label{sec: system description}
This section introduces the control system dynamics, specifications of safety requirements, and the safe control and learning problems. 

\textbf{System model.} \label{sec: system settings}
We consider a continuous-time system with the following dynamics:
\begin{align} \label{eq: sys_dynamics}
        \dot x^+_t= f(x_t)+g(x_t)u_t,
\end{align}
where $x \in \mathcal{X} \subseteq \br^d$ is the state in state space $\mathcal{X}$, $u \in \mathcal{U} \subseteq \br^p$ is the control action in action space $\mathcal{U}$, where $u_t$ as a function of $t$ is right continuous, has left hand limits at all $t$, and has only finitely many jumps in any finite interval of time. $f:\mathcal{X} \rightarrow \br^d$ and $g:\mathcal{X} \rightarrow \br^{d\times p}$ are Lipschitz continuous functions. Thus,~\eqref{eq: sys_dynamics} has a unique solution $x_t$, which is continuous in time $t$. We assume that the function $f(.)$ is completely unknown but bounded for a bounded input $x$, and the function $g(.)$ is also bounded for any bounded input $x$ and satisfies the following two assumptions. 
\begin{enumerate}[topsep=\basp,itemsep=0pt,partopsep=0pt, parsep=\insp]
    \item \label{Assumption: SVD} Assumption {1}.
    At each time $t$, matrix $g(x_t)$ admits the singular value decomposition
    \begin{equation}\label{eq: SVD of g(.)}
        g(x_t)=U_t\Sigma_tV_t^\intercal,
    \end{equation}
    where $U_t \in \br^{d \times d}$ and $V_t \in \br^{p \times p}$ are known orthogonal matrices. 
    \item \label{Assumption: Singular value bounds} Assumption {2}. Let $\lambda_{1,t}\geq\lambda_{2,t}\geq\ldots\lambda_{t,k_t}>0$ be the non-zero singular values of matrix $g(x_t)$, i.e., 
    \begin{equation}\label{eq: Singular Value Matrix}
        \Sigma_{t,(i,j)}=\begin{cases}
                            \lambda_{i,t},\text{ if }i=j\leq k_t\\
                            0,\text{ otherwise,}
                         \end{cases}
    \end{equation}
    with $k_t = rank(g(x_t)) = d$. The value of $\lambda$ is bounded as
    \begin{equation}\label{eq: singular value bounds}
        0<m_i \hat{\lambda}_{i,t} \leq \lambda_{i,t} \leq M_i \hat{\lambda}_{i,t},~~\forall i \in \{1,\ldots,k_t\},
    \end{equation}
    where $M_i, m_i, \hat{\lambda}_{i,t} > 0$ are some known positive bounds/estimates. 
\end{enumerate}




Adaptive safe control techniques such as \cite{9147463,9129764,9196709,9683085} assume $g(.)$ is fully known. Since this paper also assumes uncertainty in $g().$ along with full uncertainty in $f(.)$, the above assumptions~\ref{Assumption: SVD}  and~\ref{Assumption: Singular value bounds} about partial knowledge about $g$ are weaker than the above works. 
We go beyond such settings, and we only restrict ourselves to cases when matrices $(U_t, V_t)$ and the range and sign of $\hat{\lambda}_{i,t}$ can be inferred from the mechanics/structure of the system. For example, the sign of $\hat{\lambda}_{i,t}$ can be inferred from the fact that a bicycle (vehicle) turns left when the driving wheel is steered left, and right otherwise (see the example in appendix~\ref{sec: example} of the full paper). 

\textbf{Safety specifications.} \label{sec: safety-definitions}  A system is \textit{safe} at time $t$ when the state $x_t$ lies within a certain region $\mathcal{S} \subset \br^d$, denoted as the safe set. We characterize the safe set $\mathcal{S}$ by the level set of a continuously differentiable \textit{barrier function} $\phi : \br^d \rightarrow \br$ as follows:
\begin{align}\label{eq: CBF characterization of Safe region.}
    \mathcal{S}\triangleq\{ x: \phi(x) \geq 0,~x \in \br^d \}.
\end{align}
Additionally, we define
\begin{equation}\label{eq: safety sub-region}
    \mathcal{S}_{\theta}\triangleq\left\{x:\phi(x)\geq \theta,~\theta \geq 0\right\} \subset \mathcal{S}
\end{equation}
as the $\theta$-super-level set of $\phi$, $\partial \mathcal{S}_{\theta}\triangleq\left\{ x: \phi(x)=\theta \right\}$ as the boundary of $ \mathcal{S}$, and $int(\mathcal{S}_{\theta})\triangleq\mathcal{S}_{\theta} \setminus \partial \mathcal{S}_{\theta}$ as the interior of $\mathcal{S}_{\theta}$. In this paper, we only consider bounded safe sets $\mathcal{S}_{\theta}$ i.e., $\|x\|_2 < \infty$ for every $x \in \mathcal{S}_\theta$. The safety requirement is stated in terms of {forward invariance}, {forward convergence}, and {forward persistence} conditions.

\begin{definition}
A set $\mathcal{S}$ is said to be \textit{forward invariant} with respect to the dynamics of $\{x_t\}_t$ if $x_0 \in \mathcal{S} \Longrightarrow x_t \in \mathcal{S},~\forall t \geq 0.$ A set $\mathcal{S}$ is said to be \textit{forward convergent} with respect to the dynamics of $\{x_t\}_t$ if, given $x_0 \notin S$, $\exists \tau \geq 0~~\text{such that } x_{t} \in \mathcal{S},~~\forall t \geq \tau.$ A set $\mathcal{S}$ is said to be \textit{forward persistent} if the set $\mathcal{S}$ is both forward invariant and forward convergent. 
\end{definition}

Safe adaptive control techniques typically ensure safety after sufficiently many samples (ranging from a few to dozens) are collected (\cite{9147463,9129764,9196709,9683085}) and safe learning methods often require generative models during training or the availability of Lyapunov or barrier functions constructed from system models (~\cite{pmlr-v120-taylor20a, wagener2021safe}). The number of samples needed to learn the model usually scales with the dimension,
in contrast, this paper investigates what is fundamentally impossible or achievable for zero-violation safety and presents a method that ensures forward invariance/convergence using a minimum number of samples independent of the state dimension, for safety-critical environments. 

\textbf{Problem statement and design objectives.}
Our goal is to develop a safe adaptation method and apply it to safe learning. The problem for each setting is stated below. 
In the setting of \textit{safe adaptation}, we assume the existence of a possibly stochastic \textit{nominal controller} of the form : 
\begin{equation}\label{eq: nominal controller}
    u_t=u_{nom}( \{x_\tau \}_{\tau \leq t} , \{y_\tau \}_{\tau \leq t} ),
\end{equation}
where $u_{nom}$ is a policy that has access to the history of the state $\{x_\tau \}_{\tau \leq t}$ and observation $\{y_\tau \}_{\tau \leq t}$ and predicts a  control action $u_t \in \mathcal{U}$. We assume that the control action predicted by the nominal controller is always bounded i.e., $\|u_{nom}\|_2 < \infty$. 
Special cases of \eqref{eq: nominal controller} are memory-less controllers or the ones that only use partial information of the available history. The observation $y_\tau$ can include variables such as rewards/costs, environmental/state variables, and design parameters. Considering a nominal controller of the form \eqref{eq: nominal controller}, we do not lose any generality with respect to any specific form of information constraint. 
The nominal controller, however is not necessarily safe, due to uncertainties in the system dynamics model used to design the controller. Let 
\begin{align} \label{eq: information-history} 
    I(t,\delta)\triangleq\{ \left(x_{\tau},u_{\tau'}\right) | \tau \in \mathcal{T}_{x}(t,\delta), \tau'\in \mathcal{T}_{u}(t,\delta) \}, 
\end{align}
denote the state and action histories from the time intervals 
\begin{align}\label{eq: trajectory-history}
     \mathcal{T}_{x}(t,\delta) &\triangleq \{\tau: \tau \leq t \}\cap \{\tau: \tau \geq \max\{0,t-\delta\} \}, \text{ and} \nonumber \\ 
     \mathcal{T}_{u}(t,\delta) &\triangleq\{\tau: \tau < t\}\cap\{ \tau: \tau \geq \max\{0,t-\delta\} \},
\end{align}
respectively. 
We want to design a policy of the form:
\begin{align}\label{eq:policy structure}
    u_t = u(u_{nom}, I(t,\delta))
\end{align}
which uses information $I(t,\delta)$, and any action prescribed by the nominal controller $u_{nom}$, to produce the action $u_t$. Even when the system dynamics change, one only requires a short state and action recent history $I(t,\delta)$ that is usable for inferring safe actions. We study the design of policy \eqref{eq:policy structure} that allows for zero-violation safety with small $\delta$ to avoid causing significant interruptions to the nominal controller \eqref{eq: nominal controller}. We note that for a sufficiently fine discretization of the dynamics, a small $\delta$, and a high-dimensional state with dimension $d$, the number of samples needed by our controller $|I(t,\delta)|$, could be significantly lower than the existing methods which need at least $d$ samples to learn the dynamics before computing a safe action.


\textit{Safe learning:} We consider an application of the above algorithm where it is used as a tool in conjunction with a policy gradient algorithm similar to the REINFORCE algorithm of \cite{williams1992simple}. We discretize the time $t$ into integer multiples of a sampling interval $T_s$ such that we only observe the state variables $\{s_n\}_{n \in \mathbb N_+}$ with $s_n = x_{nT_s}$, where we use the variable $n$ , both here and in section~\ref{sec: RL applications}, to index these discrete time steps. At each time step $n$, the agent receives a reward $r(s_n,u_n)$ for the corresponding state action pair $(s_n,u_n)$. The reward function $r: \mathcal{X}\times \mathcal{U}\rightarrow \br$ is stationary with respect to time and is designed according to the environment and the desired task. The learning algorithms employ stochastic policies to perform the task, i.e., at each step $n$, an action $a_n$ is sampled from the conditional probability density function $\pi_w(.|s_n)$ i.e.,  $p(a_n|s_n) = \pi_w(a_n|s_n)$. This policy $\pi_w$ is stochastic and may depend on the history of the trajectory available at time $n$. Therefore, such a stochastic policy can be seen as a stochastic nominal controller of the form~\eqref{eq: nominal controller} with output $a_n$. This action $a_n$ is further overwritten by our policy of the form~\eqref{eq:policy structure} to get the final control action $u_n$ which is played at time $n$.
The policy is parameterized by $w \in \mathcal{W}$ and is restricted to the policy class $\Pi=\left \{ \pi_w:w \in \mathcal{W}\right\}$. For instance, in the case of a neural network model, $w$ would denote the weights of the neural network, and $\Pi$ would denote the set of such networks. Here, the goal is to maximize the objective $J(w)$ over the policy class $\Pi$ with zero-violations of safety as described above. For learning, we use the following objective function:
\begin{equation}\label{eq: Objective function}
    J(w)=\Ex_{\pi_w} \left [ \sum_{n=0}^\infty \gamma^n r(s_n,u_n) \right ],
\end{equation}
with $\gamma \in [0,1)$ as the discounting factor, $s_{n+1} \sim P(s_{n+1}|s_n,u_n),~~ \forall n $. Here $P(s_{n+1}|s_n,u_n)$ denotes the transition probability matrix of the underlying dynamics described by~\eqref{eq: sys_dynamics} under zero order hold.


\section{Safe control algorithm} \label{sec: proposed method}
In this section, we first establish an impossibility result when safety cannot be ensured due to insufficient information. Then, we propose a method to produce safe actions in an extreme situation when generative models are not available and only a handful of usable samples are given. 
Recall from section \ref{sec: system description}, that we consider system \eqref{eq: sys_dynamics} with unknown $f(.)$, and $g(.)$ satisfying assumptions \ref{Assumption: SVD} and \ref{Assumption: Singular value bounds}, and a nominal controller of the form \eqref{eq: nominal controller}. 
Without knowing $f(.)$, no algorithm can ensure zero-violation safety using the information from just one sample. This fundamental limit is formally stated below and proved in Appendix~\ref{sec: fundamental limit proof}. 
\begin{theorem}\label{thm: sample case optimality}
     Assume that there exists a state $x$ in the boundary $\partial \mathcal{S}$ of the safety set $\mathcal{S}$ such that $\nabla \phi(x) \neq 0$. Consider system \eqref{eq: sys_dynamics} with unknown $f(.)$ and known $g(.)$. Given information $I(t,\delta)$ with $\delta=0$, no policy of the form~\eqref{eq:policy structure} can ensure zero-violation safety. 
\end{theorem}



While it is impossible to ensure zero-violation safety with just one sample ($\delta=0$), under certain assumptions (assumptions~\ref{Assumption: SVD} and~\ref{Assumption: Singular value bounds}, and a continuously differentiable barrier function characterizing the safe set), there exists a policy of the form~\eqref{eq:policy structure} that guarantees zero-violation safety using as little information as $I(t,\delta)$ for any $\delta > 0$. This policy is formally presented in algorithm~\ref{alg: main algorithm} and described below. 

Algorithm~\ref{alg: main algorithm} takes as input the threshold $\theta$, and guarantees safety with respect to the safety subset $\mathcal{S}_{\theta}$. 
At each time $t$, the current state $x_t$ is observed, and $\phi(x_t)$ is computed. When the state is far from unsafe regions, \text{i.e.,} $\phi(x_t)>\theta>0$, the output of the nominal controller \eqref{eq: nominal controller} is used. Otherwise, when $\phi(x_t) \leq \theta$, actions are corrected as follows. 
Let $U_{1,t},\ldots,U_{d,t}$ denote the column vectors of the known SVD matrix $U_t$ at time $t$. Since $U_t$ is an orthogonal matrix, the vectors $U_{1,t},\ldots,U_{d,t}$ form an orthonormal basis of $\br^d$.
We compute the vector $\nabla \phi(x_t) \in \br^d$, and represent it with this new basis as:
\begin{equation}\label{eq: orthogonal decomposition of grad_phi}
    \nabla \phi(x_t)=\sum_{i=1}^d \beta_{i,t}U_{i,t},
\end{equation}
where $\beta_{i,t}$ is given by
\begin{equation}\label{eq: betas}
    \beta_{i,t}=\langle \nabla \phi(x_t), U_{i,t} \rangle,~~\forall i \in \{1,\ldots,d\}.
\end{equation}
The algorithm stores the immediate last state as $x_t^-$, and the immediate last control action as $u_{last}$ and uses them to compute $\dot x_t^-$ at each step. Then for our choice of the hyper-parameter of recovery rate $\eta>0$, we compute the parameter $\alpha_t$ as follows:
\begin{equation}\label{eq: alpha}
    \alpha_t=\frac{\langle \nabla \phi(x_t), \dot x_t^- \rangle - \eta}{\norm{\nabla \phi(x_t)}^2}.
\end{equation}
 Since the barrier function is continuously differentiable, and since $\nabla \phi (x) \neq 0,~\forall x \in \mathcal{S}$, the above $\alpha_t$ is well-defined.\\
Next we compute the matrix $\Gamma_t \in \br^{d \times d}$ which satisfies the following $d$ constraints:
\begin{equation}\label{eq: linear inequalities defining Gamma}
    \begin{aligned}
    \frac{\alpha_t \beta_{i,t}}{M_i} &> \langle U_{i,t}, \Gamma_t \dot x_t^- \rangle \text{ if } \alpha_{t} \geq 0,\beta_{t,i} \geq 0 \\ 
    \frac{\alpha_t \beta_{i,t}}{m_i} &> \langle U_{i,t}, \Gamma_t \dot x_t^- \rangle \text{ if } \alpha_{t} < 0,\beta_{t,i} \geq 0, \\
    \frac{\alpha_t \beta_{i,t}}{m_i} &< \langle U_{i,t}, \Gamma_t \dot x_t^- \rangle \text{ if } \alpha_{t} \leq 0,\beta_{t,i} < 0, \\
    \frac{\alpha_t \beta_{i,t}}{M_i} &< \langle U_{i,t}, \Gamma_t \dot x_t^- \rangle \text{ if } \alpha_{t} > 0,\beta_{t,i} < 0, 
\end{aligned}
\end{equation}

for each $i \in \{ 1,\ldots,d\}$. Observe that the algorithm~\ref{alg: main algorithm} can always find a feasible $\Gamma_t$ satisfying
  \eqref{eq: linear inequalities defining Gamma}, since there are only $d$-linear constraints on the $d^2$ entries of the matrix $\Gamma_t$. \footnote{We have used cvxpy in our numerical simulations to compute this matrix.} 
The estimated singular value matrix $\hat{\Sigma}_t$, its Moore-Penrose pseudo-inverse $\hat{\Sigma}_t^+$, matrix $\hat{g}^+(x_t)$ are given by
\begin{align}
    &\hat{\Sigma}_{t,(i,j)}  =\begin{cases}
                \hat{\lambda}_{i,t}\text{, }i=j\leq k_t\\
                        0,\text{ otherwise}
                   \end{cases}
     \hat{\Sigma}^+_{t,(i,j)} = \begin{cases}
                                    \frac{1}{\hat{\lambda}_{i,t}}\text{, }i=j\leq k_t\\
                                    0,\text{ otherwise }
                              \end{cases} \label{eq: pseudoinverse of estimated singular value matrix} \\
     &\hat{g}^+(x_t) = V_t \hat{\Sigma}_t^+ U_t^\intercal, \label{eq: g inverse estimate}
\end{align}
respectively. Finally, the control action is chosen to be
 \begin{equation}\label{eq: correction control}
    u_{corr}=u_{last}-\hat{g}^+(x_t)\Gamma_t\dot x_t^-.
 \end{equation}


\begin{algorithm}[tb]
\caption{Safety and Recovery Algorithm}
\label{alg: main algorithm}
\textbf{Input:} Barrier Function $\phi(.)$, initial state $x_{0}$, and nominal controller $u_{nom}(.)$.\\
\textbf{Hyper-parameters:} Safety margin (threshold) $\theta>0$, recovery speed $\eta>0$.\\
\textbf{Initialize:} $u_{last}$ with $ u_{nom}(x_0)$, and $x^-_{0}$ with $x_0$-the initial state.
\begin{algorithmic}[1] 
\FOR{every $t>0,$}
    \STATE Receive the current state $x_t$ from observation, and compute $\phi(x_t)$. 
    \STATE Compute the time derivative of the state variable $\dot x_t^-$.
    \IF{$\phi(x_t)>\theta$:}
        \STATE Set $u_t$ using the nominal controller \eqref{eq: nominal controller}.
    \ELSE
        \STATE Obtain $U_t$, $V_t$,  $\hat{\lambda}_{i,t}$, $M_i,m_i$ for $i=1,\dots,k_t$, and construct the singular value matrix estimate $\hat{\Sigma}_t$ using \eqref{eq: pseudoinverse of estimated singular value matrix}.
        \STATE Compute its Moore-Penrose pseudo-inverse $\hat{\Sigma}_t^+$ using \eqref{eq: pseudoinverse of estimated singular value matrix}.
        \STATE Compute the matrix $\hat{g}^+(x_t)$ using \eqref{eq: g inverse estimate}.
        \STATE Compute the parameters $\alpha_t$ and $\beta_{i,t}$ using \eqref{eq: alpha} and \eqref{eq: betas} respectively.
        \STATE Choose any matrix $\Gamma_t$ which is a feasible solution to \eqref{eq: linear inequalities defining Gamma}.
        \STATE Compute $u_{corr}$ using \eqref{eq: correction control} and assign $u_t \leftarrow u_{corr}$.
    \ENDIF
    \STATE Play $u_t$ and store $u_{last} \leftarrow u_t$, $x^-_t \leftarrow x_t$.
\ENDFOR
\end{algorithmic}
\end{algorithm}
The threshold input $\theta$ selects a subset $\mathcal{S}_{\theta}$ of the original safe set $\mathcal{S}$. So for applications which require bounded actions, depending on the bounds, the value of $\theta$ can be set to choose a conservative safety-subset so that a clipped version of the correcting action $u_{corr}$ given by~\eqref{eq: correction control} would suffice to ensure forward invariance of the original safe-set $\mathcal{S}$ in practice. 
Further, lemma~\ref{claim: main theorem claim} in Appendix~\ref{sec: section 3 proofs} shows that when $\phi(x_t) < \theta$, $\dot \phi^+ \geq \eta$. So, the parameter $\eta$ controls the recovery rate but again $\eta$ must be chosen in compliance with the physical limitations imposed by the actual controller on the magnitude of $u_{corr}$. 
One should choose these parameters to balance the trade-offs between safety margins, recovery rates, and the actual physical limitations. Algorithm~\ref{alg: main algorithm} ensures safety at all times when the state originates inside the safe set and quick recovery otherwise. These properties are formally stated below.
\begin{theorem}[Forward Persistence of algorithm~\ref{alg: main algorithm}]\label{thm: main result}
  When algorithm~\ref{alg: main algorithm} is used for the dynamical system \eqref{eq: sys_dynamics} with a known continuously differentiable barrier function $\phi$ characterizing the safety set~\eqref{eq: CBF characterization of Safe region.} with $\nabla \phi \neq 0,~\forall x \notin \mathcal{S}$, and if assumptions \ref{Assumption: SVD} and \ref{Assumption: Singular value bounds} hold, then the following are true for any value of $\theta>0$:
  \begin{enumerate}[topsep=0pt,itemsep=0pt,partopsep=0pt, parsep=\insp]
      \item When $x_0 \in \mathcal{S}_{\theta}$, the set $\mathcal{S}_{\theta}$ is forward invariant with respect to the closed loop system.
      \item If $x_0 \notin \mathcal{S}_{\theta}$, then the set $\mathcal{S}_{\theta}$ is forward convergent with respect to the closed loop system.
  \end{enumerate}
\end{theorem}
The proof of the above theorem can be found in Appendix~\ref{sec: section 3 proofs}.
An immediate consequence of Theorem \ref{thm: main result} is that algorithm~\ref{alg: main algorithm} ensures zero-violation safety by rendering the set $\mathcal{S}_{\theta}$ forward persistent with respect to the closed loop dynamics.

\section{Application to safe exploration for RL algorithms} \label{sec: RL applications}
In this section, we apply the proposed method to achieve zero-violation safety in training a policy gradient algorithm. 
We assume an unknown dynamical system \eqref{eq: sys_dynamics} satisfying assumptions~\ref{Assumption: SVD} and \ref{Assumption: Singular value bounds}.
 As in parts of section~\ref{sec: system description}, for this section, we use $n$ to index the discrete time steps of the reinforcement learning algorithms (the corresponding continuous time is $t = nT_s$, where $T_s$ is the sampling interval).
For the deterministic state dynamics given by \eqref{eq: sys_dynamics}, this corresponds to a discrete-time process where we play the state transition by adopting a first-order hold on the state variable i.e., $s_{n+1}= s_n + \dot x_{nT_S}^+ T_s.$


The integration of the proposed method is formally stated in algorithm~\ref{alg: REINFORCE with Safety}.  
Let 
\begin{equation}\label{eq: deterministic correction control}
    C(s,\dot s^-,a,u_{last})=    \begin{cases}
                            a & \text{if } \phi(s)>0\\
                            u_{corr} & \text{if } \phi(s)\leq 0. \\
                        \end{cases}
\end{equation}
 denote the control action overwritten by our technique, where $u_{corr}$ is the correction control action given by \eqref{eq: correction control} using $x_t = s_n$, and $\dot x_t^- = \dot s^- \triangleq \dot x_{(n-1)T_s}^-$.
 For algorithm~\ref{alg: REINFORCE with Safety}, the goal is to obtain the optimal stochastic policy, i.e.,
 \begin{equation}\label{eq: algorithm goal}
     w^*=\argmax_{w\in \mathcal{W}} J(w),
 \end{equation}
 such that the policy $\pi_{w^*}$ is a safe policy and the learning process occurs with a zero-violation of safety. The algorithm \textit{rolls-out} a trajectory $\ell_e=\{u_{-1},s_0,a_0,u_0,s_1,a_1,u_1,\ldots\}$ for every episode $e$ using the current stochastic policy $\pi_{w_e}$ and collects the rewards $r(s_0,u_0),r(s_1,u_1),\ldots$. 
 The trajectory roll-out continues until the system reaches a target set of states $TG$ which could be modelled as a target sub-region of the state space, $TG \subset \mathcal{X}$. Using the current policy, the log probabilities of each step in the trajectory are calculated. And the gradient $\nabla_{w_e}J(w_e)$ at each episode $e$ is estimated as:
 \begin{align} \label{eq: gradient sampling}
     \widehat{\nabla J(w_e)}  &=\left[\sum_{n=0}^\infty \nabla_{w_e}  \ln \pi_{w_e}(a_n|s_n) \right]R(\ell_e),
\end{align}
     where 
     $R(\ell_e) = \sum_{n=0}^\infty \gamma^n r(s_n,u_n)$,
 for the trajectory $\ell_e=\{u_{-1},s_0,a_0,u_0,s_1,a_1,u_1,\ldots\}$ observed in the rollout of episode $e$. 
 Here $u_n=C(s_n,\dot s_n^-,a_n,u_{n-1})$ is the overwritten control action which is actually played during the trajectory roll-out. Finally, in a way similar to  the REINFORCE algorithm by~\cite{williams1992simple}, the policy is iteratively improved after every episode by updating the weights of the neural network using stochastic gradient ascent as:
\begin{equation}\label{eq: stochastic gradient ascent}
    w_{e+1}=w_{e}+\alpha_e \widehat{\nabla J(w_e)},
\end{equation}
where $\alpha_e$ is the step size for episode $e$.
Since the correction control given by~\eqref{eq: correction control} is used to overwrite the control action every time the state reaches the boundary of safety-subset, and since Theorem~\ref{thm: main result} holds for any choice of the nominal controller, the safety set $\mathcal{S}$ stays forward invariant throughout the learning process in the limit of $T_s \rightarrow 0$. This allows for safe exploration while learning a safe controller in our setting. Since the correction is a deterministic function of the control action $a_n$, the gradient estimate computed in algorithm~\ref{alg: REINFORCE with Safety} using equation~\eqref{eq: gradient sampling} turns out to be an unbiased estimate of the true policy gradient. This is formally stated as the theorem below and proved in Appendix~\ref{sec: section 4 proofs}.

\begin{theorem}[Policy-Gradient] \label{thm: Gradient Unbiased Estimate}
Let $\widehat{\nabla J(w)}$ be the policy gradient computed by  algorithm~\ref{alg: REINFORCE with Safety} using equation~\eqref{eq: gradient sampling} and let $\nabla_w J(w)$ be the true policy gradient. Then,
\begin{equation}\label{eq: gradient theorem equation.}
    \Ex\left[ \widehat{\nabla J(w)} \right]=\nabla_w J(w).
\end{equation}
\end{theorem}


\begin{algorithm}[tb]
\caption{REINFORCE with Safety}
\label{alg: REINFORCE with Safety}
\textbf{Input:} Action space $\mathcal{U}$,  target set $TG \subset \mathcal{X}$, barrier function $\phi(.)$, the correction controller $C(s, \dot s^-,u,u_{last})$ as defined in \eqref{eq: deterministic correction control}, the safe initial state $s_{0} \in \mathcal{S}$, the estimates $U_n,V_n, \hat{\Sigma}_n$ at each instant $n$, according to assumptions~\ref{Assumption: SVD} and~\ref{Assumption: Singular value bounds}, and initial policy parameters $w_0 \in \mathcal{W}$, $u_{-1}$ any arbitrary function.\\
\textbf{Hyperparameters: }Discounting Factor $\gamma$, discretization time: $T_s$, step size: $\alpha_e>0$ for episode $e$.\\
\textbf{Initialize:} $s^-_{0}\leftarrow s_0, w \leftarrow w_0$.
\begin{algorithmic}[1] 
\FOR{every episode $e=0,1,2,\ldots$}
     \STATE Start Episode $e$.
     \STATE Set $n \leftarrow 0$. 
    \WHILE{$s_n \notin TG$} 
        \STATE Compute the time derivative of state variable:
        $\dot s_n^-\leftarrow \frac{1}{T_s}(s_n-s^-_n).$  
        \STATE Sample $a_n \sim \pi_{w_e}(.|s_n)$.  
        \STATE Overwrite the sampled action $a_n$ and obtain $u_n = C(s_n, \dot s_n^-, a_n, u_{n-1})$ 
            using \eqref{eq: deterministic correction control}. 
        \STATE Play $u_n$ and collect reward $r_n \leftarrow r(x_n,u_n)$.   
        \STATE Receive the new state $s_{n + 1}$ from observation. 
        \STATE Assign: $s^-_{n+1} \leftarrow s_n$, $n \leftarrow n+1$.
    \ENDWHILE    
     \STATE For the current episode $e$:
     let $\ell_e=(s_0,a_0,u_0,s_1,a_1,u_1,s_2,\ldots)$ be the roll-out trajectory.
     \STATE Compute $R(\ell_e)$ for the trajectory $\ell_e$ using \eqref{eq: gradient sampling}.
     \STATE Estimate the policy gradient sample $\widehat{\nabla J(w_{e})}$ using \eqref{eq: gradient sampling}.
     \STATE Update the policy network parameters using \eqref{eq: stochastic gradient ascent}.
\ENDFOR
\end{algorithmic}
\end{algorithm}

\section{Numerical Study}\label{sec: simulations}
In this section, we demonstrate effectiveness of the proposed method using numerical simulations. The simulation code can be found with the supplementary materials. 
First, we tested algorithm~\ref{alg: main algorithm} in an unknown one-dimensional linear dynamical system
with unsafe (and unstable) nominal controller starting from a safe initial state near the boundary of the safe set $\mathcal{S}$. The performance is compared with Adaptive Safe Control Barrier function (aCBF) algorithm from~\cite{9147463}; Robust Adaptive Control Barrier Function (RaCBF) and RaCBF with Set Memebership Identification (RaCBF+SMID) (denoted here as RaCBFS) algorithms from~\cite{9129764}; the convex body chasing based algorithm $\mathcal{A}_\pi-$SEL algorithm from~\cite{https://doi.org/10.48550/arxiv.2103.11055} (denoted here as cbc); and the Bayesian Learning based Adaptive Control for Safety Critical Systems (BALSA) algorithm from~\cite{9196709} (denoted here as balsa). (See appendix~\ref{sec: adaptive control simulations} for implementation details.)
Algorithm~\ref{alg: main algorithm} acted safely at all times when the state originated from a point within the safe region (forward invariance), even before other algorithms could obtain sufficient information (samples) to exhibit safe behaviors (Figure~\ref{fig: Bar Adaptive Control Plot}). This demonstrates the safety merit~\ref{Mr: Safety} of our proposed algorithm~\ref{alg: main algorithm}. When the same algorithms were started from an unsafe state outside the safe set $\mathcal{S}$, we observe a quick recovery (Figure~\ref{fig: Bar Recovery Plot}). This demonstrates the recovery merit~\ref{Mr: Recovery Extension} of our proposed algorithm~\ref{alg: main algorithm}.

Next, we tested algorithm~\ref{alg: REINFORCE with Safety} for a non-linear four-dimensional extension of the vehicular yaw dynamics system (as described in Appendix~\ref{sec: RL simulations}) with continuous state and action spaces. We compared its performance and safety against some model-free reinforcement learning algorithms like REINFORCE from~\cite{williams1992simple}, and Constrained Policy Optimization (CPO) from~\cite{pmlr-v70-achiam17a}. The policy network in algorithm~\ref{alg: REINFORCE with Safety} predicted a random action, which was then clipped within reasonable limits to obtain $a_n$ at each instance $n$.  (See appendix~\ref{sec: RL simulations} for implementation details.) We observe from figure \ref{fig: Model-Free Safety Rate}, that the average fraction of time the system was unsafe during a training epoch improves as we go from the naive REINFORCE algorithm, to the CPO algorithm but it fails to achieve zero-violation safety, which is in contrast to our algorithm~\ref{alg: REINFORCE with Safety} that achieves zero-violation safety as well.  Further from \ref{fig: Model-Free RL Convergence}, we observe that the number of steps needed to accomplish the task (after convergence) and the convergence rate of our algorithm~\ref{alg: REINFORCE with Safety} is comparable with the other two algorithms, thus showing that algorithm~\ref{alg: REINFORCE with Safety} does not suffer any compromise in its performance or convergence rate. Our tool therefore enables the REINFORCE like algorithm to learn a safe policy via a completely safe exploration process. This demonstrates the applicability to RL merit~\ref{Mr: Model-Free RL Tool} of our technique.
\begin{figure*}
     \centering
     \begin{subfigure}[b]{0.23\textwidth}
         \centering
         \includegraphics[width=\textwidth]{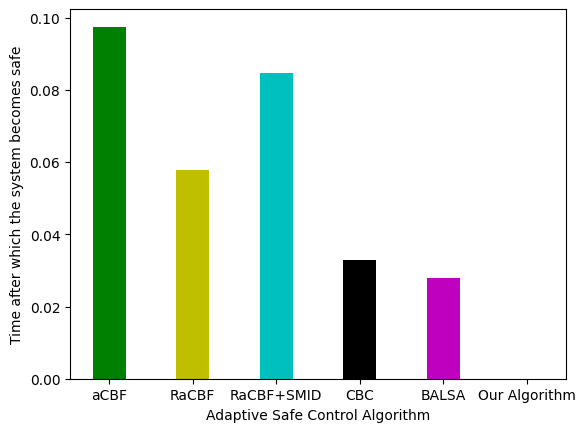}
         \caption{}
  \label{fig: Bar Adaptive Control Plot}
     \end{subfigure}
     \begin{subfigure}[b]{0.23\textwidth}
         \centering
         \includegraphics[width=\textwidth]{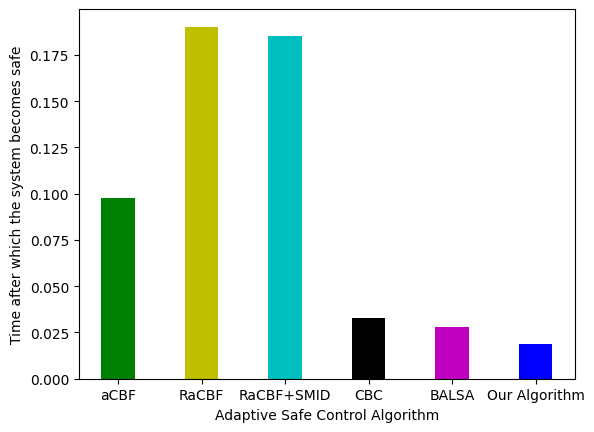}
         \caption{}
  \label{fig: Bar Recovery Plot}
     \end{subfigure}
          \begin{subfigure}[b]{0.23\textwidth}
         \centering
         \includegraphics[width=\textwidth]{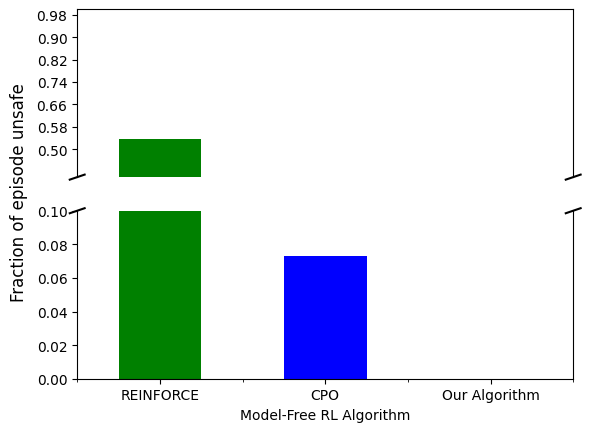}
         \caption{}
  \label{fig: Model-Free Safety Rate}
     \end{subfigure}
          \begin{subfigure}[b]{0.23\textwidth}
         \centering
         \includegraphics[width=\textwidth]{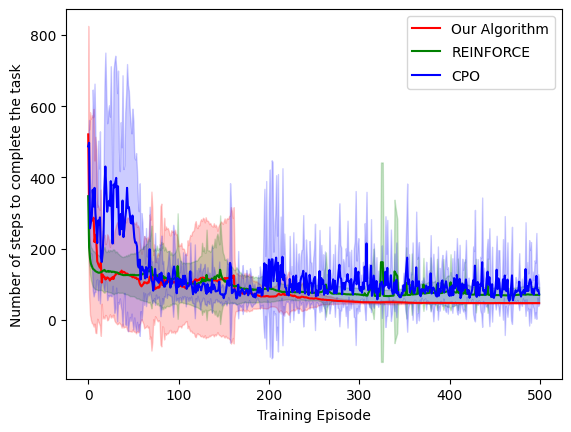}
         \caption{}
    \label{fig: Model-Free RL Convergence}
     \end{subfigure}
\caption{Comparing (a) forward invariance and (b) forward convergence of our algorithm~\ref{alg: main algorithm} with several safe adaptive control algorithms. (c) Comparing (c) safety rate and (d) convergence rate of  algorithm~\ref{alg: REINFORCE with Safety} with other model-free RL algorithms.} 
        \label{fig: Numerical study}
\end{figure*}

\section{Conclusion and Future Work}\label{sec: Paper Conclusion}
We proposed a sample-optimal technique that ensures zero-violation safety at all times for systems with large uncertainties in the system dynamics and demonstrated its applicability to safe exploration in reinforcement learning problems.
Possible future works include extensions to discrete-time systems; non-deterministic systems with an additive noise; integration into other reinforcement learning algorithms such as~\cite{https://doi.org/10.48550/arxiv.1502.05477,https://doi.org/10.48550/arxiv.1707.06347,https://doi.org/10.48550/arxiv.1509.02971}.

\section*{APPENDIX}
\section{Proof of Theorem~\ref{thm: main result}}\label{sec: section 3 proofs}



We have a few preliminary lemmas, before starting the proof of theorem~\ref{thm: main result}. We begin with the following lemma which is a extension of the well-known Nagumo's theorem~\cite{Nagumo1942berDL} to the settings when the dynamics $x_t$ does not need to be differentiable w.r.t. $t$, but only requires the existence of left and right hand derivatives w.r.t. $t$, and we present an independent elementary proof of the lemma below.
\begin{lemma}\label{lem: Nagumo using RI}
Let $h:\br^d \rightarrow \br$ be a differetiable function, and let $x_{t}:\br_{\geq 0}\rightarrow \br^d$ denote the state-variable of a dynamical system which is continuous with respect to time $t$ and such that $\dot x^+_t$ exists $\forall t$. Let $\mathcal{C}\triangleq\left\{ x:h(x)\geq 0  \right\},$ and let $\partial \mathcal{C}\triangleq\left\{ x: h(x)=0\right\}$ be the boundary of the set $\mathcal{C}$, $int(\mathcal{C})\triangleq\mathcal{C}\setminus\partial \mathcal{C}$ be the interior of $\mathcal{C}$, and let $T>0$ be any finite time-horizon. If the following conditions hold:
\begin{description}
    \item \textit{S-1:} $x_0 \in int(\mathcal{C})$.
    \item \textit{S-2:} $x_t \notin int(\mathcal{C}) \implies \dot h^+(x_t) > 0$.
\end{description}
Then
\begin{equation}\label{eq: Nagumo's implication}
    x_t \in \mathcal{C},~~\forall t \in [0,T].
\end{equation}

\end{lemma}

We need a couple of lemmas to prove lemma~\ref{lem: Nagumo using RI}.
 We begin with the following lemma which is a straightforward extension of the chain rule for differentiation to right hand derivatives.

\begin{lemma}\label{lem: chain rule for right hand derivatives}
If $\phi:\br^d\rightarrow \br$ is differentiable with respect to $x$, $\forall x \in \br^d$, if $x_t$ is continuous with respect to time $t$ such that $\dot x^+_t$ exists at all times $t$, then
\begin{equation}\label{eq: dot phi positive}
    \dot \phi^+(x_t)=\nabla \phi(x_t) \cdot \dot x^+_t.
\end{equation}
\end{lemma}

\begin{proof}[Proof (Lemma \ref{lem: chain rule for right hand derivatives})]
Since $\phi(x)$ is differentiable with respect to $x$, we define $Q(y,x_t)$ for any $y \in \br^d$ as:
\begin{equation*}
    Q(y,x_t)=\begin{cases}
                \begin{pmatrix}
                    \frac{\phi(y)_1-\phi(x_t)_1}{y_1-{x_t}_{1}} \\
                    \vdots\\
                    \frac{\phi(y)_i-\phi(x_t)_i}{y_i-{x_t}_{i}} \\
                    \vdots\\
                    \frac{\phi(y)_d-\phi(x_t)_d}{y_d-{x_t}_{d}} \\
                \end{pmatrix}
                \text{ if } y\neq x_t,\\
                \nabla \phi(x_t) \text{ otherwise.}
             \end{cases}
\end{equation*}
We observe that $Q(y,x_t)$ is continuous with respect to $y$ due to the differentiability of $\phi$ at $x_t$. And since $x_\tau$ is continuous with respect to $\tau$, so $Q(x_\tau,x_t)$ is also continuous with respect to $\tau$. Now let $\tau=t+\delta t$, for some $\delta t>0$. Then, we have:
\begin{align*}
    \dot \phi^+(x_t)&=\lim_{\delta t \rightarrow 0^+}
    \frac{\phi(x_{t+\delta t})-\phi(x_t)}{\delta t}\\
    &=\lim_{\tau \rightarrow t^+} Q(x_{\tau},x_t)\cdot \frac{1}{\tau-t}(x_\tau-x_t)\\
    &=\lim_{\tau \rightarrow t^+}Q(x_{\tau},x_t)\cdot\lim_{\tau \rightarrow t^+}\frac{1}{\tau-t}(x_\tau-x_t)\\
    &=Q(x_t,x_t)\cdot \dot x^+_t\\
    &=\nabla \phi(x_t)\cdot \dot x^+_t,
\end{align*}
where the first line uses the definition of right hand derivatives, the second line follows from the definition of $Q(x_\tau,x_t)$, the third and fourth line use the continuity of $Q(x_\tau,x_t)$ with respect to $\tau$, the existence of $\dot x^+_t$, and the fact that the limit of a product of two functions is the product of the limits of the two functions when the limits exist, and the last line again follows from the definition of $Q(y,x_t)$.
\end{proof}
Similarly, we can also prove the following:
\begin{lemma}\label{lem: chain rule for left hand derivatives}
    If $\phi:\br^d\rightarrow \br$ is differentiable with respect to $x, \forall x \in \br^d$, if $x_t$ is continuous with respect to time $t$ such that $\dot x^-_t$ exists at all times $t$, then
\begin{equation}\label{eq: dot phi negative}
    \dot \phi^-(x_t)=\nabla \phi(x_t) \cdot \dot x^-_t.
\end{equation}
\end{lemma}
The next lemma is another result required in the proof of lemma~\ref{lem: Nagumo using RI}. It resembles a direct consequence of the Lagrange's mean value theorem for differentiable functions if the function were differentiable.

 \begin{lemma}\label{lem: LMVT Consequence for RHD}
     Let $h(t)$ be a continuous function of $t$ with its right-hand derivative $\dot h^+$ defined for every $t$. Let $t_1<t_2$ be such that $h(t_1)>h(t_2)$, then there exists a $t_3 \in (t_1,t_2)$ such that $\dot h^+(t_3) \leq 0$. 
 \end{lemma}
 \begin{proof}
     For contradiction, assume 
     \begin{equation}\label{eq:RHD contradiction}
         \dot h^+(\tau)>0,~~\forall \tau \in (t_1,t_2).
     \end{equation}
    We partition the interval $[t_1,t_2]$ into $k$ equal intervals:
    $[t_{i,1},t_{i,2}]$ for $i=1,2,\ldots,k$, such that $t_{1,1}=t_1$, $t_{k,2}=t_2$, and $t_{i+1,1}=t_{i,2}$, and $t_{i,1}+t_{i+1,2}=2t_{i,2}$, $\forall i \in  \{1,2,\ldots, k-1\}$. Thus, we have
    $$[t_1,t_2]=\bigcup_{i=1}^{k} [t_{i,1},t_{i,2}],~~\text{where }[t_{i,1},t_{i,2}] \cap [t_{j,1},t_{j,2}]=\emptyset~~\forall i\neq j.$$ 
    Since $h(t_1)>h(t_2)$, there exists $i \in \{1,2,\ldots,k\}$ such that $h(t_{i,1})>h(t_{i,2})$, otherwise $h(t_1)=h(t_{1,1})\leq h(t_{1,2})=h(t_{2,1})\leq h(t_{2,2})=\ldots=h(t_{k,2})=h(t_2)$, a contradiction. For this $i$, we define $f_1=\dot h^+(t_{i,1})$. By our contrary assumption \eqref{eq:RHD contradiction}, $f_1>0$.
    Next we partition this interval $[t_{i,1},t_{i,2}]$ again into $k$ equal parts exactly as above (note that we again have $h(t_{i,1})>h(t_{i,2})$) to get the number $f_2$. We keep repeating the process to generate the sequence: $f_1,f_2,\dots$. We observe that each number $f_j$ in the above sequence is positive as argued above, and since the function $h(t)$ is right-differentiable everywhere in $(t_1,t_2)$ and the sequence is constructed by sub-partitioning the previous partition within $[t_1,t_2]$, the sequence $\{f_l\}_{l=1}^\infty$ converges to $\dot h^+(t_3)$ for some $t_3 \in (t_1,t_2)$. Further, since each $f_l$ is generated from a partition interval $[t_{i(l)}, t_{j(l)}]$, such that $h(t_{i(l)}) > h(t_{j(l)})$, and the function $h$ is right hand differentiable, we must have $\lim_{l\rightarrow \infty} f_l \leq 0$.  So, we have 
    \begin{equation}\label{eq: LMVT RHD conclusion}
        \dot h^+(t_3) = \lim_{l\rightarrow \infty} f_l  \leq 0,
    \end{equation}
    which contradicts \eqref{eq:RHD contradiction} since $t_3 \in (t_1,t_2)$. Thus, we have the desired lemma.
 \end{proof}
 We are now ready to prove lemma~\ref{lem: Nagumo using RI}.
 \begin{proof}[Proof (Lemma~\ref{lem: Nagumo using RI})]
     Since $h(x)$ is differentiable with respect to $x$ and $x_t$ is continuous in $t$, $h(x_t)$ is continuous with respect to $t$. Since $\dot x_t^+$ is given by~\eqref{eq: sys_dynamics}, then from lemma~\ref{lem: chain rule for right hand derivatives}, we have:
     $\dot h^+(x_t)=\nabla h(x_t) \cdot \dot x_t^+$ exists for every $t \in [0,T]$.
     \par \noindent Assume for contradiction, that
     \begin{equation}\label{eq: Nagumo's contradiction}
         \exists \tau \in [0,T]~~\text{such that } h(x_\tau)<0.
     \end{equation}
     Now from \textit{S-1}, we have $h(x_0)>0$. So, we have $h(x_0)>0>h(x_{\tau})$. Since $h(x_t)$ is continuous in $t$, by the intermediate value theorem, there exists a $ t^* \in (0,\tau)$ such that $h(x_{t^*})=0$ and $h(x_{t'})<0$ for every $t'\in (t^{*},\tau).$
     So, we have $x_{t'} \notin int(\mathcal{C})$ for every $t' \in [t^*,\tau]$. Now \textit{S-2} implies, 
     \begin{equation}\label{eq: positive RHD from S-2}
         \dot h^+(x_t)>0,~~\forall (t^*,\tau).
     \end{equation}
     Now setting $t_1=t^*$, $t_2=\tau$, in lemma~\ref{lem: LMVT Consequence for RHD} and since $h(x_{t^*})=0>h(x_{\tau})$, we have from lemma~\ref{lem: LMVT Consequence for RHD}, that there exists $\tau' \in (t^*,\tau)$ such that:
     \begin{equation}\label{eq: Nagumo's new proof conclusion}
         \dot h^+(x_{\tau'}) \leq 0,
     \end{equation}
     which contradicts \eqref{eq: positive RHD from S-2}. Thus, we have the desired lemma.     
 \end{proof}

 Next we have the following lemmas regarding the properties of algorithm~\ref{alg: main algorithm}.

 \begin{lemma}\label{lem: u_last lemma}
     When algorithm~\ref{alg: main algorithm} is used for dynamical system~\eqref{eq: sys_dynamics}, we have:
     \begin{align}\label{eq: left hand derivative of state}
         \dot x^-_t = f(x_t) + g(x_t)u_{last}
     \end{align}
 \end{lemma}
 \begin{proof}[Proof (Lemma~\ref{lem: u_last lemma})]
        Since $u_t$ in~\eqref{eq: sys_dynamics} has left hand limits, we have $\lim_{\tau \rightarrow t^-} u_\tau = u_{last}$. Further, since $u_t$ can have only finitely many jumps in any finite interval, and let $\delta t > 0$ be sufficiently small such that $u_t$ is continuous in $[t-\delta t, t]$. Then, we have for any $ \tau \in [t-\delta t, t]$ using~\eqref{eq: sys_dynamics}, $x_\tau = x_{t-\delta t} + \int_{t-\delta t}^\tau  \left ( f(x_s) + g(x_s) u_s \right ) ds $. Now using continuity of $f(x_t)$, $g(x_t)$, and $x_t$ with respect to $t$:
        \begin{align*}
            \dot x^-_t &= \lim_{\delta t \rightarrow 0^+} \frac{x_{t} - x_{t-\delta t}}{\delta t}\\
            &= \lim_{\delta t \rightarrow 0^+} \frac{\int_{t-\delta t}^{t} \left ( f(x_s) + g(x_s) u_s \right )}{\delta t}\\
            &= f(x_t) + g(x_t)\lim_{s \rightarrow t^-} u_s\\
            &= f(x_t) + g(x_t)u_{last}.
        \end{align*} 
 \end{proof}

 \begin{lemma}\label{lem: Bounded u}
     Let algorithm~\ref{alg: main algorithm} be used for the dynamical system~\eqref{eq: sys_dynamics}, with a nominal controller $u_{nom}$ which is finite at any time $t$. Then the control action $u_t$ played by algorithm~\ref{alg: main algorithm} is always finite at any finite time $t$. 
 \end{lemma}

 \begin{proof}[Proof (Lemma~\ref{lem: Bounded u})]
     From~\eqref{eq: sys_dynamics}, we have a unique solution $x_t$ which is continuous in time $t$. So, the state variable $x_t$ is finite at any finite time $t$. Now at any point $t$, the control action $u_t$ applied by algorithm~\ref{alg: main algorithm} is either the finite control action $u_{nom}$ provided by the nominal controller or the correction control action $u_{corr}$ given by~\eqref{eq: correction control}. Now using an inductive argument, if $u_\tau$ is finite for all $\tau < t$ (vacuously true for $t=0$), then $u_{last}$ is finite. Further, from continuity of $f$ and $g$, and~\eqref{eq: left hand derivative of state} from lemma~\ref{lem: u_last lemma}, we have a finite $\dot x^-_t$ at any finite time $t$. Further, by~\eqref{eq: g inverse estimate} and assumption~\ref{Assumption: Singular value bounds}, $\hat{g}^+(x_t)$ is also finite at any finite time $t$. From~\eqref{eq: linear inequalities defining Gamma}, we also observe that $\Gamma_t$ can be chosen as a matrix with finite entries at any finite time $t$. Thus, by equation~\eqref{eq: correction control}, $u_{corr}$ is finite at any finite time $t$. Thus, in both cases the control action $u_t$ applied by algorithm~\ref{alg: main algorithm} is finite for any finite time $t$.
 \end{proof}


 \begin{lemma}\label{claim: main theorem claim}
    When algorithm~\ref{alg: main algorithm} is used for dynamical system \eqref{eq: sys_dynamics}, the following holds for any time $t$:
    $$ \phi(x_t)\leq \theta \implies \dot \phi^+(x_t)\geq \eta >0.$$
\end{lemma}

\begin{proof}[Proof (Lemma~\ref{claim: main theorem claim})]
Let $t$ be such that $\phi(x_t)\leq \theta$. Then algorithm~\ref{alg: main algorithm} plays $u = u_{corr}$ given by~\eqref{eq: correction control}. So, we have:
\begin{align}
    \dot \phi^+(x_t)&=\nabla \phi(x_t) \cdot \dot x_t^+\label{eq: derivation 1}\\
    &= \nabla \phi(x_t)\cdot \left(f(x_t)+g(x_t)u_{corr}\right)\label{eq: derivation 2}\\
    &= \nabla \phi(x_t) \cdot \left(f(x_t)+g(x_t)u_{last}-g(x_t)\hat{g}^+(x_t)\Gamma_t\dot x_t^-\right)\label{eq: derivation 3}\\
    &= \nabla \phi(x_t) \cdot \left( \dot x_t^- -g(x_t)\hat{g}^+(x_t)\Gamma_t\dot x_t^-\right)\label{eq: derivation 4}\\
    &= \dot \phi^-(x_t)-\nabla \phi(x_t)\cdot \left( g(x_t)\hat{g}^+(x_t)\Gamma_t\dot x_t^- \right)\label{eq: derivation 5}
\end{align}
where \eqref{eq: derivation 1} uses \eqref{eq: dot phi positive}, \eqref{eq: derivation 2} uses \eqref{eq: sys_dynamics}, \eqref{eq: derivation 3} uses \eqref{eq: correction control}, \eqref{eq: derivation 4} uses \eqref{eq: left hand derivative of state}, and \eqref{eq: derivation 5} uses \eqref{eq: dot phi negative}.
Now let 
\begin{equation}\label{eq: y_t}
    y_t=\Gamma_t\dot x_t^-.
\end{equation}
Then, from \eqref{eq: SVD of g(.)}, \eqref{eq: Singular Value Matrix}, \eqref{eq: pseudoinverse of estimated singular value matrix} and \eqref{eq: g inverse estimate}, we have:
\begin{equation}\label{eq: intermediate matrix products}
    g(x_t)\hat{g}^+(x_t)y_t=\sum_{i=1}^{k_t}\frac{\lambda_{i,t}}{\hat{\lambda}_{i,t}}\langle U_{i,t},y_t \rangle U_{i,t}.
\end{equation}
Now using \eqref{eq: y_t}, and \eqref{eq: orthogonal decomposition of grad_phi} in \eqref{eq: intermediate matrix products}, and the orthonormality of the $U_{i,t}$ vectors we have:
\begin{equation}\label{eq: intermediate sum}
    \nabla \phi(x_t)\cdot \left ( g(x_t)\hat{g}^+(x_t)\Gamma_t\dot x_t^- \right )=\sum_{i=1}^{k_t}\beta_{i,t}\frac{\lambda_{i,t}}{\hat{\lambda}_{i,t}}\langle U_{i,t},y_t \rangle.
\end{equation}

Now, we define $\epsilon_{i,t}$, for every $i\in \{1,\ldots,k_t\}$ as:
\begin{equation}\label{eq: epsilon definition}
    \epsilon_{i,t}\triangleq \alpha_t \beta_{i,t} - \frac{\lambda_{i,t}}{\hat{\lambda}_{i,t}}\langle U_{i,t}, y_t \rangle.
\end{equation}
So we have, 
\begin{align}
    \dot \phi^+(x_t)&=\dot \phi^-(x_t)-\nabla \phi(x_t)\cdot \left( g(x_t)\hat{g}^+(x_t)\Gamma_t\dot x_t^- \right) \nonumber\\
    &=\dot \phi^-(x_t)-\sum_{i=1}^{k_t} \beta_{i,t}\left ( \alpha_t\beta_{i,t}-\epsilon_{i,t} \right )\label{eq: derivation 6}\\
    &=\dot \phi^-(x_t)-\alpha_t \sum_{i=1}^{k_t} \beta_{i,t}^2 + \sum_{i=1}^d \beta_{i,t}\epsilon_{i,t}\label{eq: derivation 7}\\
    &=\dot \phi^-(x_t)-\langle \nabla \phi(x_t),\dot x_t^-\rangle + \eta + \sum_{i=1}^{k_t} \beta_{i,t}\epsilon_{i,t} \label{eq: derivation 8}\\
    &= \eta + \sum_{i=1}^{k_t} \beta_{i,t}\epsilon_{i,t},\label{eq: derivation 9}
\end{align}
where we get \eqref{eq: derivation 6} from substituting \eqref{eq: intermediate sum},\eqref{eq: epsilon definition} in \eqref{eq: derivation 5}, we get \eqref{eq: derivation 8} using \eqref{eq: betas}, \eqref{eq: alpha} and the fact that $u_{i,t}$'s are orthonormal and $k_t = d$, which implies $$\langle \nabla \phi (x_t), \dot x^-_t \rangle - \eta= \alpha_t \norm{\nabla \phi(x_t)}^2=\alpha_t\sum_{i=1}^{d} \beta_{i,t}^2 = \alpha_t\sum_{i=1}^{k_t} \beta_{i,t}^2,$$ and we get \eqref{eq: derivation 9} using \eqref{eq: dot phi negative}.

Now, in order to prove $\dot \phi^+(x_t)\geq \eta >0$, it suffices to show that:
\begin{equation}\label{eq: constraint in terms of epsilon}
    \sum_{i=1}^{k_t} \beta_{i,t}\epsilon_{i,t} \geq 0.
\end{equation}
Substituting for $\epsilon_{i,t}$ from \eqref{eq: epsilon definition} and $y_t$ from \eqref{eq: y_t} in the above constraint \eqref{eq: constraint in terms of epsilon}, we get the following constraint on the matrix $\Gamma_t$:
\begin{equation}\label{eq: constraint on Gamma Marix}
    \sum_{i=1}^{k_t} \beta_{i,t} \left ( \alpha_t\beta_{i,t}-\frac{\lambda_{i,t}}{\hat{\lambda}_{i,t}}U_{i,t}^\intercal \Gamma_t \dot x_t^- \right ) \geq 0.
\end{equation}
Now, using the bounds on the estimates of the singular values of $g(x_t)$ from \eqref{eq: singular value bounds}, we observe that any solution $\Gamma_t$ to the set of $d$ inequalities in \eqref{eq: linear inequalities defining Gamma} would make each of the $d$ terms in~\eqref{eq: constraint on Gamma Marix} non-negative. Thus, our algorithm's choice of $\Gamma_t$ satisfies the overall constraint~\eqref{eq: constraint on Gamma Marix}, thereby implying $\dot \phi(x_t)^+ \geq \eta$ (from \eqref{eq: derivation 9} and \eqref{eq: constraint in terms of epsilon}), and proving the claim. 
\end{proof}

 We are now ready to prove the main result of this section, theorem~\ref{thm: main result}, using lemma~\ref{lem: Nagumo using RI}, lemma~\ref{lem: chain rule for right hand derivatives} and lemma~\ref{claim: main theorem claim}.

\begin{proof}[Proof (Theorem~\ref{thm: main result})]

We set $h(x_t)=\phi(x_t)-\theta$, $\mathcal{C}=\mathcal{S}_{\theta}$ in lemma~\ref{lem: Nagumo using RI}. 
In case of forward invariance, we have $x_0 \in \mathcal{S}_{\theta}$. This makes condition \textit{S-1} of lemma~\ref{lem: Nagumo using RI} true for our choice of $h$ and $\mathcal{C}$.
Since $\phi(x)$ is differentiable with respect to $x$ and $x_t$ is continuous in time $t$, $\phi(x_t)$ is continuous in $t$. Thus, we have $h$ satisfying the conditions of lemma~\ref{lem: Nagumo using RI}. Further, from lemma~\ref{lem: Bounded u}, $u_t$ applied by algorithm~\ref{alg: main algorithm} is finite. Then by equation~\eqref{eq: sys_dynamics} we conclude that $\dot x_t^+$ always exists.
Thus, lemma~\ref{lem: Nagumo using RI} is applicable for our choice of $h$ and $\mathcal{C}$ and lemma~\ref{claim: main theorem claim} guarantees that \textit{S-2} of lemma~\ref{lem: Nagumo using RI} holds true as well. Thus, by lemma~\ref{lem: Nagumo using RI} the set $\mathcal{S}_\theta$ is forward invariant as desired.

Now using lemma~\ref{claim: main theorem claim}, the forward convergence of $\mathcal{S}_{\theta}$ (part $2$ of theorem~\ref{thm: main result}) can also be guaranteed. In particular, we show a forward convergence rate for $\mathcal{S}_{\theta}$ as follows. Let
\begin{equation}\label{eq: distance to safety sub-region}
    d(x_t,\mathcal{S}_{\theta})=\theta-\phi(x_t)>0.
\end{equation}
 Since $x_0 \notin \mathcal{S}_{\theta}$, from lemma~\ref{claim: main theorem claim} we have $\dot \phi^+(x_t)\geq \eta$ whenever $\phi(x_t)<\theta$. This uniform lower bound on $\dot \phi^+$ implies (using mean value theorem on $\phi$), that within time: 
\begin{equation}\label{eq: forward convergence time}
    \tau \leq \frac{d(x_t,\mathcal{S}_{\theta})}{\eta},
\end{equation}
the system reaches the safety region $\mathcal{S}_{\theta}$. Also, for a finite $x_t$ by lemma~\ref{lem: Bounded u}, the control action played by algorithm~\ref{alg: main algorithm}: $u_t$ is bounded. So, by equation~\eqref{eq: sys_dynamics}, within finite duration $\tau$, the state variable $x_t$ remains bounded throughout. So, by using lemma~\ref{lem: Bounded u} in an inductive way, we observe that the state variable $x_t$ as well as the applied control action $u_t$ always stays finite (bounded) in the case of forward convergence as well.
We thus conclude the proof of theorem~\ref{thm: main result}.
\end{proof}

\section{Proof of Theorem~\ref{thm: sample case optimality}}\label{sec: fundamental limit proof}
\begin{proof}[Proof (Theorem~\ref{thm: sample case optimality})]
At time $t=0$, let $x_0 \in \partial\mathcal{S}$, such that $\nabla \phi(x_0)\neq 0$. From the assumption in theorem~\ref{thm: sample case optimality}, we know such a state $x_0 \in \partial \mathcal{S}$ exists. For the above choice of $x_0$, we have $I(0,0)=\{x_0\}$ at $t=0$. Then, we show that any controller observing just the sample $x_0$ cannot guarantee the forward invariance of $\mathcal{S}$.  
Let the controller play $u_0$ at time $t=0$. Choose any $f(.)$ which satisfies $f(x_0)=-\nabla \phi(x_0)-g(x_0)u_0$. We the have:
$$ \dot \phi^+(x_0)=\nabla \phi(x_0) \cdot (f(x_0)+g(x_0)u_0)=-\norm{\nabla \phi(x_0)}^2<0,$$
where the first equality uses~\eqref{eq: sys_dynamics}, lemma~\ref{lem: chain rule for right hand derivatives}, and the above choice of $f(.)$, and the inequality follows since $\nabla \phi(x_0)\neq 0$ for our choice of $x_0$. Further, since $x_0 \in \partial \mathcal{S}$, $\dot \phi(x_0)<0$ implies there exists an  $\epsilon$ such that for any $\tau \in (0,\epsilon)$, $\phi(x_\tau)<0$, i.e. $x_\tau \notin \mathcal{S}$ and thus the algorithm fails to guarantee forward invariance of $\mathcal{S}$ in this interval. Since the function $f(.)$ is unknown to the controller, our choice of the $f(.)$ above is valid, and we can construct such $f$ for every controller $u$. We have thus shown the existence of a system corresponding to that particular controller whose dynamics equation is of the form~\eqref{eq: sys_dynamics}, where the controller fails to guarantee the forward invariance of $\mathcal{S}$ with respect to the system. This proves theorem~\ref{thm: sample case optimality}.\\

\end{proof}

\section{Proof of Theorem~\ref{thm: Gradient Unbiased Estimate}}\label{sec: section 4 proofs}

\begin{proof}[Proof (Theorem~\ref{thm: Gradient Unbiased Estimate})]

Let $\mathcal{H}$ be the set of possible state action trajectories indexed by $e$. Let $\ell_e \in \mathcal{H}$ denote the state action trajectory observed in the roll-out for episode $e$,
\begin{equation}\label{eq: episode trjacetory}
    \ell_e=\{ s_n, a_n, u_n\}_{n=0}^\infty,
\end{equation}
and let 
\begin{equation} \label{eq: trajectory net reward}
    R(\ell_e)=\sum_{n=0}^\infty \gamma^n r(s_n,u_n)
\end{equation}
 denote the discounted sum reward for trajectory $\ell_e=\{s_n,a_n,u_n\}_{n=0}^\infty \in \mathcal{H}$. 

Let $P_w$ denote the probability density function of trajectories induced by the correction controller \eqref{eq: correction control} and the current stochastic policy $\pi_{w}$. By construction, for any trajectory $\ell_e=\{s_0,a_0,u_0,s_1,a_1,u_1,\ldots\}$, $P_w(\ell_e)$  given by 
\begin{align}\label{eq: probability distribution induced for trajectory}
    P_w(\ell_e) &= \pi_w(a_0|s_0)\Pr[u_0=C(s_0,\dot s_{0}^-,a_0,u_{-1})]\Pr[s_1|s_0,u_0] \nonumber \\ 
    & \pi_w(a_1|s_1)\Pr[u_1=C(s_1,\dot s_{1}^-,a_1,u_0)]\Pr[s_2|s_1,u_1]\ldots,
\end{align}
is the induced probability density function supported on the set of possible trajectories $\mathcal{H}$. Further, since $C(s_n,\dot s_n^-, a_n, u_{n-1})$ is a deterministic function given by \eqref{eq: correction control}, so we have $ \Pr [u_n = C(s_n,\dot s_n^-, a_n, u_{n-1})] = 1,~\forall n$ in the above.
Therefore, we can write:     
\begin{align}\label{eq: PG derivation}
    &\nabla_w J(w) = \nabla_w \Ex_{\ell_e \sim P_w} R(\ell_e) \nonumber\\ 
    &=\nabla_w \int_{\ell_e \in \mathcal{H}}R(\ell_e) P_w(\ell_e) d\ell_e \nonumber\\
    &=\int_{\ell_e \in \mathcal{H}} R(\ell_e) \nabla_w P_w(\ell_e) d\ell_e \nonumber\\
    &=\int_{\ell_e \in \mathcal{H}} P_w\left (\ell_e \right )  \left( \frac{\nabla_w P_w(\ell_e)}{P_w(\ell_e)}R(\ell_e)\right ) d\ell_e \nonumber \\
    &=\Ex_{\ell_e \sim P_w} \nabla_w \ln P_w(\ell_e) R(\ell_e) \nonumber \\
    &= \Ex_{\ell_e \sim P_w} R(\ell_e) \nabla_w 
     \{  \ln \pi_w(a_0|s_0) \nonumber \\
     &+ \ln \Pr[u_0=C(s_0,\dot s^-_{0},a_0,u_{-1})] 
     +\ln \Pr[s_1|s_0,u_0] \nonumber \\
     &+ \ln \pi_w (a_1|s_1) + \ldots \}  \nonumber \\ 
    &= \Ex_{\ell_e \sim P_w} \sum_{n=0}^\infty\left [ \nabla_w \ln \pi_w (a_n|s_n) \right ] R(\ell_e) \nonumber\\
    &= \Ex_{\ell_e \sim P_w} \widehat{\nabla J(w)}.
\end{align}
Here, in the second line we exchanged the integral with the gradient using the Leibniz integral rule, and the second to last line uses the fact that $\Pr[s_n|s_{n-1},u_n],~C(s_n,\dot s_n^-, a_n, u_{n-1})$ are independent of the policy parameters $w$. In particular, $\ln \Pr [C(s_n,\dot s_n^-, a_n, u_{n-1})] = 0,~\forall n$, and is independent of $w$.  So, the gradient of these log probabilities with respect to $w$ is zero. From \eqref{eq: PG derivation}, we see that the gradient estimate $\widehat{\nabla J(w)}$ used by algorithm~\ref{alg: REINFORCE with Safety} is indeed an unbiased estimate of the true policy gradient $\nabla_w J(w)$.
\end{proof}


\section{Bicycle dynamics example illustrating the applicability of assumptions~\ref{Assumption: SVD} and~\ref{Assumption: Singular value bounds}}\label{sec: example}

In this section, we show an example of bicycle dynamics where our assumptions~\ref{Assumption: SVD} and~\ref{Assumption: Singular value bounds} are applicable to a reasonable partial model of the dynamics.\\
We consider a special case of our assumed setting where the matrix $g(x_t)$ always admits an SVD with known constant matrices $U_t=U$ and $V_t=V$, and when some gross upper and lower bounds on the singular values are available. For example, in vehicular yaw dynamics of a bicycle, if we only have the knowledge of axis lengths, mass and moment of inertia about the vertical axis ( quantities easier to measure ), then it is sufficient to apply our proposed method since the setting falls under our assumptions~\ref{Assumption: SVD} and~\ref{Assumption: Singular value bounds} as demonstrated below.
The two-dimensional yaw dynamics of a bicycle is typically modeled with state $x_t = \left [ v(t)~r(t) \right ]^\intercal$ as:  
  \begin{equation}\label{eq: bicycle dynamics}
      \begin{bmatrix}
          \dot v(t)\\
          \dot r(t)
      \end{bmatrix} 
      =\begin{bmatrix}
          -\frac{c_r+c_f}{mu} & \frac{c_ra_2-c_fa_1}{mu}-u\\
          \frac{c_ra_2-c_fa_1}{uI_z} & -\frac{c_fa_1^2+c_ra_2^2}{uI_z}
      \end{bmatrix}
      \begin{bmatrix}
          v(t)\\
          r(t)
      \end{bmatrix}
      +\begin{bmatrix}
            \frac{c_f}{m} \\
            \frac{c_fa_1}{I_z}
       \end{bmatrix} \delta_f,
  \end{equation}
where $m$ is the vehicle mass; $a_1,a_2$ are the distances to the center of mass point of the front and rear wheels respectively; $I_z$ is the vehicle's moment of inertia about the $z$-axis (the bicycle moves in the $x$-$y$ plane); $r$ is the angular velocity (anti-clockwise) of the turn being made by the bicycle steering; $u$ is the bicycle's velocity along the $x$-axis (direction of its initial motion); $v$ is the velocity along $y$-axis (direction along which it is turning); $c_f,c_r$ are the unknown cornering stiffness constants of the front and rear wheels; and $\delta_f$ is the only control action which is the angle of left turn applied through the handle (\cite{guiggiani2014science}). Here, the matrix $g(.)$ is constant and admits the following singular value decomposition:
\begin{equation}
    \begin{bmatrix}
        \frac{c_f}{m}\\
        \frac{c_fa_1}{I_z}
    \end{bmatrix}=
    \begin{bmatrix}
       \frac{\lambda_1}{\lambda} & \frac{\lambda_2}{\lambda} \\
       \frac{\lambda_2}{\lambda} & - \frac{\lambda_1}{\lambda}
    \end{bmatrix}
    \begin{bmatrix}
        \lambda\\
        0
    \end{bmatrix}
    \begin{bmatrix}
        1
    \end{bmatrix},
\end{equation}
where $\lambda=\sqrt{(c_f/m)^2+(c_fa_1/I_z)^2}$ is the singular value $\lambda_1=\frac{c_f}{m}$, and $\lambda_2=\frac{c_fa_1}{I_z}$. Now if the stiffness parameters $c_f,c_r$ are unknown (difficult to measure in a realistic setting), then $\lambda_1$, $\lambda_2$ would be unknown scalars, but the stiffness parameters would get cancelled out in the the normalized fractions $\frac{\lambda_1}{\lambda}$ and $\frac{\lambda_2}{\lambda}$ so these fractions would still be known to the user. Also, since $\lambda$ is strictly positive, any rough estimates on the unknown stiffness parameters would allow the user to guess a $\hat{\lambda}$ with $0<m\hat{\lambda}<\lambda<M\hat{\lambda}$, where $m$ and $M$ are known (ratio) bounds. Thus, such a system satisfies our assumptions~\ref{Assumption: SVD} and~\ref{Assumption: Singular value bounds}. 

\section{Implementation Details of Numerical Simulations from section \ref{sec: simulations}}

\subsection{Comparison with Adaptive Safe Control Algorithms}\label{sec: adaptive control simulations}
We consider the simple one dimensional system below:
\begin{align}\label{eq: Simulation 1D System}
    \dot x =1.5x +u,
\end{align}
where $f(x)=1.5x$, $g(x)=1$, with a nominal controller: $u_{nom}(x)=-x$. We define the safe set as $\mathcal{S}=[-0.2,0.2]$, and the initial state is $x_0=0.1999$. 
We use the same shorthand for denoting each algorithm as listed in the main paper section~\ref{sec: simulations}.
\begin{enumerate}
    \item \textbf{[Algorithm~\ref{alg: main algorithm}]} Here, for our algorithm we use the nominal controller $u_{nom}=-x$, which makes the closed loop dynamics:
    \begin{equation*}
      \dot x=0.5x,
    \end{equation*}
    which is exponentially increasing with time and is expected to quickly exit the safe set $\mathcal{S}$. We choose the barrier function $\phi(x)=1-25x^2$, corresponding to our safe set $\mathcal{S}$ and we take $\theta=0.001$, and $\Gamma=4$.
    
    \item \textbf{[aCBF]} For the Adaptive CBF algorithm~\cite{9147463}, we have:\\
     $f(x)=x$, $g(x)=1$, $u_{nom}(x)=-x$, $F(x)=x$, $\theta^*=0.5$,
     $\Gamma =5$, $\hat{\theta}_{0}=0$, $\phi_a(x,\hat{\theta})=1-25x^2.$ This corresponds to the open loop dynamics~\ref{eq: Simulation 1D System}.\\
     This gives us an adaptation rule of:
     \begin{align}\label{eq: adaptation rule}
         \dot{\hat{\theta}}&=\Gamma \tau(x,\hat{\theta}),\\
         \textrm{where } \tau(x,\hat{\theta})&=-F(x)\frac{\partial \phi_a}{\partial x}=50x^2.
     \end{align}
        And for the $\lambda_{cbf}$ function we have:
        \begin{equation}\label{eq: lambda cbf}
            \lambda_{cbf}(x,\hat{\theta})=\hat{\theta}-\Gamma\frac{\partial \phi_a}{\partial \hat{\theta}}=\hat{\theta}.
        \end{equation}
        So, we have the $aCBF$ controller  given by the following quadratic program:
        \begin{align*}
            \text{minimize}_{u \in \br} ~~ & \norm{u-u_{nom}(x)}\\
            \text{subject to}~~ & \sup_{u \in \br} (-50x[x+\hat{\theta}x+u]) \geq 0.
        \end{align*}
    \item \textbf{[RaCBF]} For the RaCBF algorithm~\cite{9129764}, we have the same setting as above on the same system \eqref{eq: Simulation 1D System} and with the same $\lambda_{cbf}$ \eqref{eq: lambda cbf} and adaptation rule \eqref{eq: adaptation rule}. But here, we assume that the error in parameter estimation: $\Tilde{\theta}=\theta^*-\hat{\theta} \in [-\Tilde{\nu},+\Tilde{\nu}]$ (i.e. parameter estimation error lies in a bounded convex region). Here we have taken $\Tilde{\nu}=2$. We also took $\alpha(x)$ as the identity function: $\alpha(x)=x$, so we have the following quadratic program for the RaCBF controller:
        \begin{align*}
            \text{minimize}_{u \in \br} ~~ & \norm{u-u_{nom}(x)}\\
            \text{subject to}~~ & \sup_{u \in \br} (-50x[x+\hat{\theta}x+u])\\ &\geq -\phi_a(x,\hat{\theta})+\frac{2}{5}.
        \end{align*}
    \item \textbf{[RaCBFS]} This algorithm referred to as the RaCBF+SMID algorithm in~\cite{9129764} improves upon the previous algorithm (RaCBF) of the same paper, by updating the uncertainty bound $\Tilde{\nu}$ periodically after every few instances. The algorithm, updates its uncertainty over $\theta^*$ so that it is consistent with the trajectory history so far. Depending on the hyperparameter $D$ (we have taken $D=0.07$), we set $r_{min}$ and $r_{max}$ such that at time $t$: 
    \begin{align*}
        \forall \theta \in [r_{min},r_{max}]:~~~~~~~~&\\
         |\dot{x}_{\tau}-f(x_\tau)-g(x_\tau)u_\tau-F(x_\tau)\theta|&\leq D\\
         ~~\forall \tau \in [0,t).~~~~~~~~~&
    \end{align*}
    Now with this uncertainty quantification over $\theta^*$ given by $r_{min},r_{max}$, the algorithm updates $\Tilde{\nu}$ every $2.5$ms (where the time horizon is 0.25s) as follows:
    \begin{equation*}
        \Tilde{\nu}=\max\left \{|r_{min}-\hat{\theta}|,|r_{max}-\hat{\theta}|\right \}.
    \end{equation*}
    
    \item \textbf{[cbc]} This is the convex body chasing algorithm from~\cite{https://doi.org/10.48550/arxiv.2103.11055}. Here we search for consistent models corresponding to a pair $(\alpha_t,\beta_t)$ at each round $t$ which serve as estimates for the actual parameters $(\alpha^*,\beta^*)$ governing the system dynamics as:
    $\dot x=\alpha^* x+\beta^* u$, with initial uncertainties taken as $|\alpha^*|\leq 5$, and $2.5 \times 10^{-6} \leq \beta^* \leq 0.05$. Discretizing the system dynamics for a small sampling time $T_s=2.5\times 10^{-4}s$, we look for all candidates $(\alpha,\beta)$ consistent with the current trajectory history as the polytope $P_t \subset P_{t-1}$ given by the linear program corresponding to the constraints: $|\alpha x_i+\beta u_i-x_{i+1}|\leq \eta$, for all discrete points $i$ in history up to time $t$, where $\eta$ is a suitably chosen hyperparameter. At every discrete step $t$, then the candidate $(\alpha_t,\beta_t) \in P_t$ is chosen as the Euclidean projection of the previous candidate $(\alpha_{t-1},\beta_{t-1})$ on the current polytope: $P_t$. At every round $t$, once the candidate model parameters $(\alpha_t,\beta_t)$ are found the control action $u_t$ is then chosen as: 
    $u_t=-\frac{\alpha_t}{\beta_t}x_t$. The loss function $\mathcal{G}(x_t,u_t)$ is taken as the indicator function, $\mathcal{G}(x_t,u_t)=\mathbf{1}\{x_t \notin \mathcal{S}\}.$
    
    \item \textbf{[balsa]} This is the Bayesian learning for adaptive safety algorithm from~\cite{9196709}. Here a proxy $\hat{f}(x)$ for the function $f(x)$ is available to the controller and it has full knowledge of $g(x)$.
    From the pseudo control $\mu$, it computes the control action as:
    \begin{align}
        u&=g^{-1}(x)(\mu-\hat{f}(x)),\text{ where }\\
        \mu&=\mu_{rm}+\mu_{pd}+\mu_{qp}-\mu_{ad}.
    \end{align} 
    
    Here, the modelling error $\Delta(x)=f(x)-\hat{f}(x)$ is estimated by a Bayesian Neural Network estimator that takes as input $x$ and predicts $m,s$ as the mean and standard deviation of the prediction $\Delta(x)$. The neural network is trained on the trajectory data history after every $0.025s$, where the total horizon is $0.25s$. From this we sample $\mu_{ad}\sim \mathcal{N}(m,s^2)$ for the adaptive part of the pseudo control to cancel the modelling error. For this setting, we have taken $\mu_{rm}=0$ (since there is no tracking here), $\mu_{pd}(x)=\dot x -g(x)u_{last}+g(x)u_{nom}(x)=\mu_{nominal}(x)$, which is the effective pseudo control for our nominal controller used in~\ref{alg: main algorithm}. Finally the pseudo control $\mu_{qp}$ for safety is obtained by solving the quadratic program:
        \begin{align*}
            \text{minimize}_{\mu_{qp},d \in \br} ~~ & \mu_{qp}^2+5d^2\\
            \text{subject to:}~~ & \phi_0 + \phi_1\mu_{qp}\geq d\\
            &d \leq 0,
        \end{align*}
        where $\phi_0=\nabla \phi(x) \cdot \mu_{pd}(x)$, and $\phi_1=\nabla \phi(x)$.
\end{enumerate}
When $x_0 = 0.1999$, the forward invariance data was observed as shown by figure~\ref{fig: Bar Adaptive Control Plot}, and when $x_0=0.2500$, a forward convergence (recovery) was observed as shown by figure~\ref{fig: Bar Recovery Plot}.
\subsection{Comparison with RL Algorithms}\label{sec: RL simulations}

We consider a 4-dimensional extension of the non-linear vehicular yaw dynamics \eqref{eq: bicycle dynamics} given by with state $x_t = [ V_y, r, \psi, y] ^\intercal $ and the dynamics:

\begin{equation}\label{eq: 4-D vehicular dynamics}
    \dot x_t = \begin{bmatrix}
        \frac{-c_0}{mV} & \frac{-c_1}{mV} -V & 0 & 0\\
        \frac{-c_1}{I_z V} & \frac{-c_2}{I_z V} & 0 & 0\\
        0 & 1 & 0 & 0\\
        \cos \psi & 0 & V_x \frac{\sin \psi}{ \psi } & 0
    \end{bmatrix} \begin{bmatrix}
        V_y \\
        r \\
        \psi \\
        y
    \end{bmatrix}
    +\begin{bmatrix}
        \frac{C_{\alpha f}}{m}\\
         \frac{a C_{\alpha f}}{I_z}\\
         0\\
         0
    \end{bmatrix} \delta,
\end{equation}
where $V_y$ is the velocity of the vehicle along the lateral direction to which it is pointing, $V_x$ is the velocity along the direction the vehicle is headed, $V = \sqrt{V_x^2 + V_y^2}$ is the net magnitude of the vehicle's velocity, $r = \dot \psi$, is the yaw-rate i.e., the clockwise angular velocity along which the vehicle is turning in the $x-y$ plane (horizontal plane, the vehicle intends to take a right turn facing north to facing east eventually), $\psi$ is the yaw i.e., the current heading angle (clockwise) it makes from North towards East directions, and $y$ is the lateral displacement of the vehicle from the starting position (ideally we don't want the vehicle to move much and take a successful right turn without displacing too much). $m$ is the mass of the vehicle, $I_z$ is the moment of inertia of the angle along the vertical axis (direction of gravity) about which the vehicle rotates to take the turn, and $a$ is the distance between the center of the mass of the vehicle to front wheel. The control variable is the scalar variable $\delta$, which is the steering angle applied by the controller.

We assume all the state variables are completely observable, $C_{\alpha f}, m$ are known and $a, I_z$ are known only upto an unknown scaling factor. The parameters $c_0, c_1, c_2, $ depend on several drag and stiffness parameters which are harder to model/estimate.   So, the first matrix (the product of the $4 \times 4$ matrix with the state vector) is the $f(x)$ vector which can be completely unknown to the user, the second matrix ($4 \times 1$ matrix) is the $g(x)$ matrix, where the SVD becomes:
\begin{align} \label{eq: vehicular dynamics 4-D SVD}
    & [C_{\alpha f}/m, a C_{\alpha f}/I_z, 0, 0 ]^\intercal =\\
    & \begin{bmatrix}
    \frac{1/m}{\sqrt{1/m^2 + a^2/I_z^2}} & \frac{-a/I_z}{\sqrt{1/m^2 + a^2/I_z^2}} & 0 & 0 \\
    \frac{a/I_z}{\sqrt{1/m^2 + a^2/I_z^2}} & \frac{1/m}{\sqrt{1/m^2 + a^2/I_z^2}} & 0 & 0\\
    0 & 0 & 1 & 0 \\
    0 & 0 & 0 & 1     
    \end{bmatrix}
    \begin{bmatrix}
    \lambda \\
    0\\
    0\\
    0        
    \end{bmatrix}
    \begin{bmatrix}
    C_{\alpha f} 
    \end{bmatrix}^\intercal ,
\end{align}
where $\lambda = \sqrt{\frac{1}{m^2}+\frac{a^2}{I_z^2}}$.
In the above SVD, the matrices $U$, $V$ are completely known since $m$, $C_{\alpha f}$ are known and the ratio $a/I_z$ is also known. Further, the singular value $\lambda$ is also known for this example. In order to demonstrate the effectiveness of our method, in our simulation we only used an estimate of the singular value $\lambda$ within known factors of $0.2$ and $5$. In particular, we took:  
$m=100,~I_z=20,~a=1,~m=0.2,~M=5,~C_{\alpha f} = 10 $. For the simulation, we took $c_0 = 70$, $c_1= 40$, $c_2 = 180$, and an initial net velocity $V = 5$. All the state variables were initialized at zero, and since the applied control may increase the velocity and yaw rates, they were clipped between $\pm 7$ and $\pm 350$ units respectively while running the dynamics, so as to keep all quantities within a realistic range. Further, the control variable  applied according to the nominal controller (policy network) \eqref{eq: nominal controller}, or the applied correction control \eqref{eq: correction control} was always clipped within $\pm 100$ to stay within realistic controller limitations (such bounds are not within assumptions~\ref{Assumption: Singular value bounds} and~\ref{Assumption: SVD}, but our algorithm still works for this simulation). 

For the reinforcement learning algorithms, we have used a policy network with just two hidden layers of $100$ neurons each, the input layer takes the $4$-dimensional state variable as input and the output layer predicted the mean value of the control action. The action chosen by the reinforcement learning algorithm was a Gaussian random variable as predicted by the policy network and standard deviation $0.7$. The discretization interval was $T_s = 0.02$. All the three algorithms were trained for $500$ episodes and each episode consisted of an integer number of steps until it succeeded to achieve the right turn within a maximum of $1000$ steps (if the agent failed to achieve the goal within $1000$ steps it failed that episode).  All the simulations were done for a set of $10$ experiments and the averaged data within $pm$ standard deviation are plotted in figures~\ref{fig: Model-Free RL Convergence} as the convergence rate plot. The averaged data over these experiments regarding the safety fraction (or percentage) is plotted in figure~\ref{fig: Model-Free Safety Rate}.

The reward was set as 
\begin{align}\label{eq: RL simulation reward}
    &r(s_n,a_n,s_{n+1}) = \nonumber\\
    &\begin{cases}
    -4 + \frac{1}{4}(\frac{1}{(\psi_{n} - \pi/2)^2+0.0001}) \text{ if } |\psi_{n+1} - \pi/2| > \pi/36\\
    7000 \text{ if } |\psi_{n+1} - \pi/2|<\pi/36 \text{ and terminate if this occurs. }
    \end{cases}
\end{align}
Similarly, the cost (penalizing unsafe behavior) for the CPO algorithm benchmark was defined as:
\begin{align}\label{eq: RL simulation cost}
    & c(s_n,a_n,s_{n+1}) = \nonumber \\
    & \min\{ (0.004 (r_{n+1}- 50\pi)^2 + 10^{-6} (V_{y_{n+1}} - 2.5)^2),~0.1 \}.
\end{align}
And the barrier function corresponding to our safety based correction controller was set as:
\begin{equation}\label{eq: RL simulation barrier funcion}
    \phi(x) = 200 - 4(r-50\pi)^2 - 0.001(V_y - 2.5)^2.
\end{equation}
For the correction controller algorithm~\ref{alg: REINFORCE with Safety}, we used an input threshold of $\theta = 500$, and the recovery rate hyper-parameter $\eta = 500$.

    The simple REINFORCE like algorithm from~\cite{williams1992simple} does not care about safety and learns to perform the task in about 47-50 steps after around 400 training episodes. Our algorithm which overwrites the REINFORCE like algorithm's prediction using our correction controller also achieves a similar convergence rate. To compare our result against a safe RL benchmark, we compared it against the
 CPO algorithm~\cite{pmlr-v70-achiam17a} with the single constraint:
$$\sum\limits_{n=0}^\infty \gamma^n c(s_n,a_n,s_{n+1}) \leq 0.01,$$ 
with a cost function given by \eqref{eq: RL simulation cost}. In all the three algorithms, \eqref{eq: RL simulation reward} was used for rewards to learn the optimal policy. The CPO benchmark algorithm learns the safe policy quicker in about 100-150 episodes but the optimal learnt policy takes about 80 steps to perform the task as shown by figure~\ref{fig: Model-Free RL Convergence}. For all three algorithms, a discounting factor of $\gamma = 0.99$ was used.
As figure~\ref{fig: Model-Free Safety Rate} shows even the safe RL benchmark was not entirely safe during the learning phase and the naive REINFORCE like algorithm was very often unsafe compared to our algorithm~\ref{alg: REINFORCE with Safety} which was always safe and achieved zero-violation safety along with learning.



\bibliographystyle{IEEEtran}
\bibliography{references}


    

\end{document}